\documentclass[journal]{IEEEtran}

\usepackage{cite}
\usepackage{amsmath,amssymb,amsfonts,amsthm}

\usepackage[ruled,linesnumbered]{algorithm2e}
\usepackage{algorithmicx}
\usepackage{algpseudocode}

\usepackage{graphicx}
\usepackage{textcomp}
\usepackage{array}
\usepackage{booktabs}
\usepackage{makecell}
\usepackage{color}

\usepackage{subfig}
\usepackage[margin=0.2cm]{caption}

\usepackage{stfloats}
\usepackage{url}
\usepackage{verbatim}
\usepackage{balance}

\def\BibTeX{{\rm B\kern-.05em{\sc i\kern-.025em b}\kern-.08em
		T\kern-.1667em\lower.7ex\hbox{E}\kern-.125emX}}

\newtheorem{theorem}{Theorem}

\newtheorem{definition}{Definition}

\begin{document}

\title{CRT: Collision-Tolerant Residence Time for Deterministic Transmission in LEO Satellite Networks}

\author{
  Siqi Yang, 
  Zonghui Li, 
  Chaoqun You, 
  Yue Gao, ~\IEEEmembership{Fellow,~IEEE}
\thanks{S. Yang, C. You and Y. Gao are with the Institue of Space Internet, Fudan University, Shanghai 200438, China, and the College of Computer Science and Artificial Intelligence, Fudan University, Shanghai 200438, China (email: sqyang23@m.fudan.edu.cn, chaoqunyou@fudan.edu.cn, gao.yue@fudan.edu.cn). Z. Li is with Beijing Key Lab of Traffic Data Analysis and Mining, School of Computer Science and technology, Beijing Jiaotong University, Beijing 100044, China (email: lizonghui@bjtu.edu.cn). The corresponding authors are Z. Li and C. You. }
}

\maketitle

\begin{abstract}
	Low-Earth Orbit (LEO) satellite networks are a key enabler for the 6G Non-Terrestrial Network (NTN) architecture. However, supporting time-sensitive services in LEO networks is challenging due to highly dynamic topologies and the difficulty of maintaining precise global time synchronization. Existing Time-Sensitive Networking (TSN) mechanisms largely rely on static topologies and strict  synchronization, which makes them ill-suited to dynamic LEO environments. To address this issue, we propose CRT, a deterministic transmission framework tailored for LEO networks. CRT regulates per-hop residence time using local clocks, thereby compensating for link-delay variations without requiring strict global synchronization. To handle asynchronous collisions, CRT adopts a collision-tolerant scheduling strategy that maximizes the number of schedulable flows while bounding collision-induced jitter. We formalize the corresponding scheduling problem and show that it is NP-hard. We further develop CRT-Fast, an efficient heuristic algorithm. It combines iterative layering with path continuity to control collision intensity and improve path stability under topology changes. Simulations on Iridium and Starlink constellations show that the proposed method achieves lower delay jitter and high schedulability under heavy traffic loads.
\end{abstract}


\section{Introduction}
Low-Earth-Orbit (LEO) satellite networks are gaining attention for global connectivity due to their broad coverage and high availability. Driven by the miniaturization of hardware and reduced launch expenditures, massive LEO constellations, such as Starlink\cite{starlink}, Iridium\cite{iridium}, and OneWeb\cite{Oneweb}, are proliferating to bridge the digital divide in isolated and underserved regions. 
Concurrently, the transition toward 6G frameworks\cite{popovski2022perspective,yu2023toward} has imposed increasingly stringent requirements on LEO infrastructures. Time-critical services such as tele-surgery, autonomous transport, and industrial automation require strictly bounded latency and near-zero packet loss.

To fulfill these exacting requirements, LEO networks must facilitate the deterministic transmission of data flows\cite{wang2024learning,sun2025joint}.

Nevertheless, the inherent volatility of LEO constellations presents a formidable obstacle to achieving deterministic communication. The high-velocity orbital motion of satellites results in transient link unavailability and frequent handovers\cite{xiao2022leo, ma2023network,cao2023satcp}. Consequently, both Inter-Satellite Links (ISLs) and Ground-Satellite Links (GSLs) suffer from fluctuating propagation delays and intermittent connectivity. Such instability complicates the precise scheduling of time-sensitive traffic and undermines deterministic service guarantees. This leads to a fundamental challenge: maintaining deterministic performance within a highly fluid topology. This challenge must be addressed before LEO networks can support mission-critical time-sensitive services.

Deterministic networking technologies have recently garnered considerable interest  \cite{tian2024large,wang2021large,hu2024time}. Distinguishing itself from traditional best-effort routing that calculates paths on a hop-by-hop basis in real time, these technologies emphasize the pre-establishment of end-to-end (e2e) routes and the advance reservation of network resources. Time-Sensitive Networking (TSN) \cite{finn2018introduction,nasrallah2018ultra}, standardized by the IEEE 802.1 working group, has emerged as a primary technical framework for providing such guarantees.

The TSN framework integrates gate-controlled scheduling, time synchronization, preemption, and seamless redundancy. This allows Ethernet to support deterministic delivery for delay-critical traffic while maintaining compatibility with best-effort services. Rooted in the Time-Triggered (TT) communication paradigm\cite{kopetz2003time}, TSN aligns end systems and switches through precise time synchronization before executing transmissions at predefined instants. Network-wide synchronization is typically managed via IEEE 802.1AS\cite{IEEE8021AS}, which builds upon the IEEE 1588 precision clock protocol\cite{IEEE1588}. Furthermore, IEEE 802.1Qbv\cite{IEEE8021Qbv} introduces the Time-Aware Shaper, utilizing Gate Control Lists (GCLs) to regulate packet forwarding at high-precision intervals.

However, satellite links differ significantly from terrestrial TSN environments due to continuous and rapid propagation-delay variations, which can reach tens of milliseconds\cite{pan2023measuring} and induce substantial jitter. Furthermore, attaining nanosecond-level time synchronization across multi-satellite LEO systems is exceptionally difficult. Discrepancies in propagation delay can cause signal misalignment and simultaneous interference, hindering both synchronous reception and accurate peer-delay estimation. Factors such as orbital perturbations and the non-spherical geometry of the Earth further aggravate synchronization errors\cite{chen2024asynchronous}. Consequently, terrestrial deterministic methods, which rely on stable topologies and negligible sync errors, may become suboptimal or entirely ineffective when deployed in LEO constellations.

To address these limitations, we propose CRT, a novel scheduling mechanism for deterministic transmission in LEO satellite networks. CRT addresses two fundamental challenges: (1) delay jitter caused by highly dynamic satellite topologies, and (2) limited time synchronization among local clocks in LEO satellites. To address them, CRT combines three key strategies.

First, instead of relying on globally synchronized transmission instants, CRT computes the residence time of each data flow at every switch node based on local timestamps. In this way, it stabilizes the e2e delay without requiring precise global synchronization. Second, recognizing that strict non-overlapping routing limits network capacity, CRT adopts a collision-tolerant strategy. It maximizes flow schedulability while minimizing link overlap. By calculating the Worst-Case Delay (WCD) caused by asynchronous collisions, CRT ensures that the resulting jitter remains bounded within the deadline margin. Third, to cope with topology dynamics, CRT introduces a seamless handover mechanism based on implicit backup paths and path continuity, thereby mitigating packet loss during topology changes.
CRT enables deterministic transmission with high schedulability in dynamic LEO environments without requiring precise global synchronization.
The major contributions of this paper are summarized as follows:

\begin{itemize}
	\item \textbf{Deterministic transmission mechanism without global synchronization.}
	We propose CRT, a residence-time-based deterministic transmission framework for LEO satellite networks. By regulating per-hop residence time using local timestamps, CRT compensates for link-induced dynamics without relying on strict global clock synchronization.
	
	\item \textbf{Collision-tolerant scheduling model with bounded jitter.}
	We analyze the impact of asynchronous clock drift on shared-link transmissions and show that collisions among flows from distinct sources are unavoidable in unsynchronized LEO environments. Based on this observation, we formulate the CRT scheduling problem, which maximizes schedulability while minimizing the overlap degree to bound collision-induced jitter under deadline and resource constraints. We further show that this problem is NP-hard.
	
	\item \textbf{Practical heuristic scheduling algorithm.}
	To solve the CRT scheduling problem efficiently, we develop CRT-Fast, a heuristic algorithm that combines iterative layering with path continuity. CRT-Fast controls collision intensity during flow scheduling while improving stability across time slots.
	
	\item \textbf{Comprehensive evaluation on representative LEO constellations.}
	We conduct extensive simulations on Iridium and Starlink constellations. The results show that CRT-Fast achieves better schedulability, lower jitter, stronger path stability under topology dynamics, and good scalability compared with baseline methods.
\end{itemize}

\section{Background}
To facilitate the understanding of the proposed CRT mechanism, we first introduce the basic terminology and the deterministic transmission principles in TT networks.

\begin{figure}[h]
	\centering
	\includegraphics[width=\linewidth]{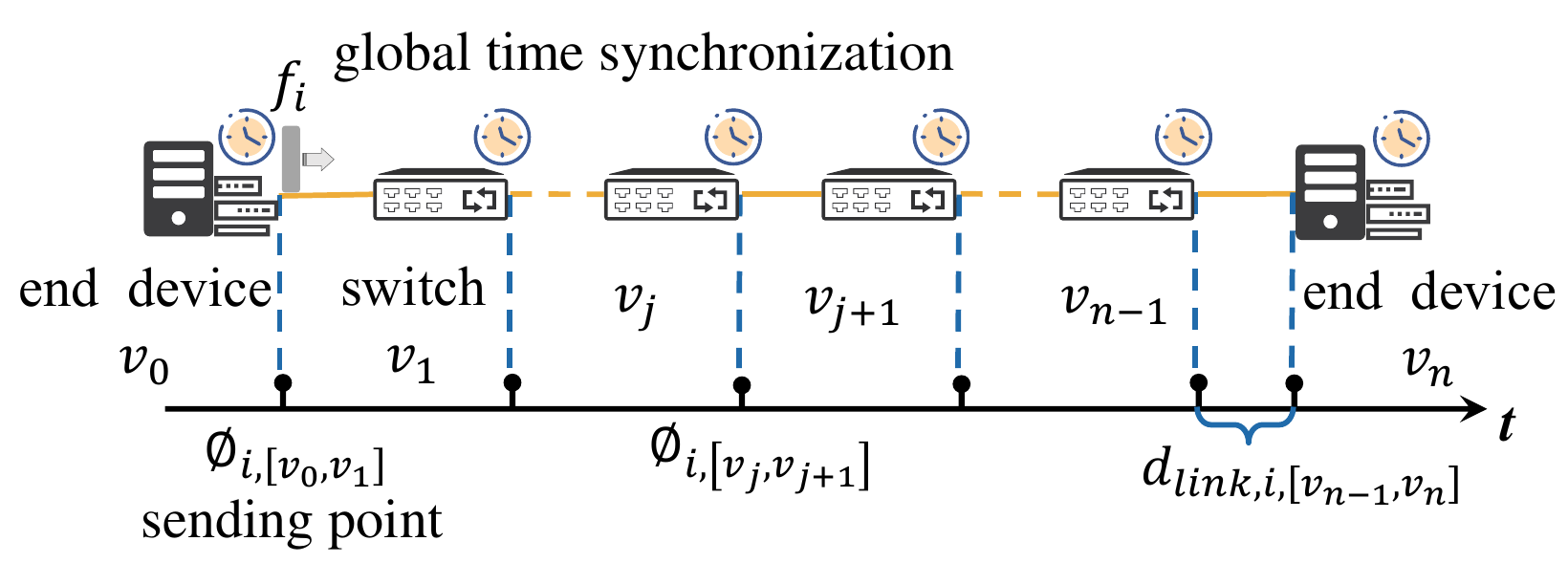}
	\caption{Example of a deterministic data flow path in TT networks.}
	\label{TTNetworkExample}
\end{figure}

\subsection{Network Topology and Dataflow}
The physical topology is represented as an undirected graph $G(V,E)$, where $V$ consists of end systems and switches, and $E$ represents the physical communication links. Each physical link defines two directed dataflow links. We define the set of directed links $\mathcal{L}$ as:

\begin{equation}
	\mathcal{L}=\{[u,v],[v,u]\mid \{u,v\}\in E\}.
\end{equation}

A dataflow path $p$ is defined by a sequence of directed links from a sender to a receiver. For example, as illustrated in Fig. \ref{TTNetworkExample}, the path from sender $v_0$ to receiver $v_n$ is:
\begin{equation}
	p = \langle [v_0, v_1], [v_1, v_2], \ldots, [v_{n-1}, v_n] \rangle.
\end{equation}
A TT flow $f_i$ is a periodic data stream characterized by its period $T_i$ and frame length $L_i$. In a standard synchronized network, the transmission of $f_i$ on a specific link $[u, v]$ is triggered at a precise, pre-scheduled time instant, known as the offset, denoted by $\phi_{i, [u,v]}$.

\subsection{Deterministic Transmission}

TT networks depend on global time synchronization, which establishes a unified temporal reference for deterministic transmission. TT schedulers typically use linear constraint models that incorporate factors such as e2e delay bounds, conflict-free transmission requirements, and buffer constraints. The goal is to determine precise sending instants for TT flows at both end devices and intermediate switches. The computed schedules are compiled into static schedule tables \cite{li2018time}, which are then deployed to the devices. As a result, each device transmits TT flow frames strictly according to the specified time instants. For example, consider the $n$-th departure time of flow $f_i$ from vertex $v_j$ to vertex $v_{j+1}$, which is specified as:
\begin{equation}
	t_{depart} = n \cdot T_i + \phi_{i, [v_j, v_{j+1}]},
\end{equation}
where $T_i$ is the period, and $\phi_{i, [v_j, v_{j+1}]}$ denotes the scheduled offset on the link.

Alongside TT flows, the network accommodates non-TT traffic. IEEE 802.1Qbv \cite{IEEE8021Qbv} introduces guard bands to prevent transmission conflicts. These bands precede each TT transmission. They allow non-TT frames to complete transmission before TT frames begin. However, this guard band is set to the maximum frame length of non-TT traffic, leading to potential bandwidth waste. IEEE 802.1Qbu \cite{IEEE8021Qbu} introduces a preemption mechanism to address this. It allows TT frames to interrupt ongoing non-TT transmissions. This mechanism reduces the guard band to 127 bytes, which corresponds to the maximum non-preemptible frame length. Based on this, we assume that TT frames can be transmitted immediately at their scheduled departure times. There is no delay due to ongoing transmissions.

As a result, the arrival time of flow $f_{i}$ at $v_{n}$ is determined as follows:
\begin{equation}
	Arr_{i, v_{n}} = \phi_{i, [v_{n-1}, v_{n}]} + d_{link,i, [v_{n-1}, v_{n}]},
\end{equation}
where $d_{link,i, [v_{n-1}, v_{n}]}$ denotes the link delay on $[v_{n-1}, v_{n}]$. This delay is dynamically measured using the peer delay mechanism defined in IEEE 1588 \cite{IEEE1588}. Accordingly, the e2e delay of flow $f_{i}$ from $v_{0}$ to $v_{n}$ is given by:
\begin{equation}
	Delay_{i}^{e2e} = Arr_{i, v_{n}} - \phi_{i, [v_{0}, v_{1}]}.
\end{equation}

However, maintaining consistent transmission schedules in TT networks requires system-wide time synchronization. The synchronization accuracy is denoted by $\mu$. It represents the maximum allowable time offset between any two synchronized devices. As a result, the e2e delay jitter of flow $f_{i}$ falls within the range $Delay_{i}^{e2e} \pm \mu$. This characterizes the deterministic latency guarantee offered by TT networks.

While TT-based networks provide high predictability under static, well-synchronized conditions, their underlying assumptions are often violated in the LEO environment. Their reliance on fixed schedule tables, global time alignment, and rigid transmission instants makes them ill-suited for topologies characterized by frequent connectivity shifts and uncertain propagation delays. These constraints emphasize the urgent need for a more resilient and adaptive scheduling paradigm.

\section{The Residence Time Mechanism and Collision Challenges}
\label{sec:mechanism}
\subsection{Residence Time Mechanism}
To address the challenges of dynamic topologies and asynchrony, we propose the Residence Time mechanism. As shown in Fig. \ref{ReTi_Transmission}, the proposed mechanism differs from traditional methods relying on the global time synchronization, it
stabilizes the e2e delay by regulating the per-hop residence
time according to local clocks. Let $\Delta t_{i,v}$ denote the allocated residence time of flow $f_i$ at node $v$ before it is transmitted to the next hop. 

\begin{figure*}[htbp]
	\centering
	\includegraphics[width=0.78\textwidth]{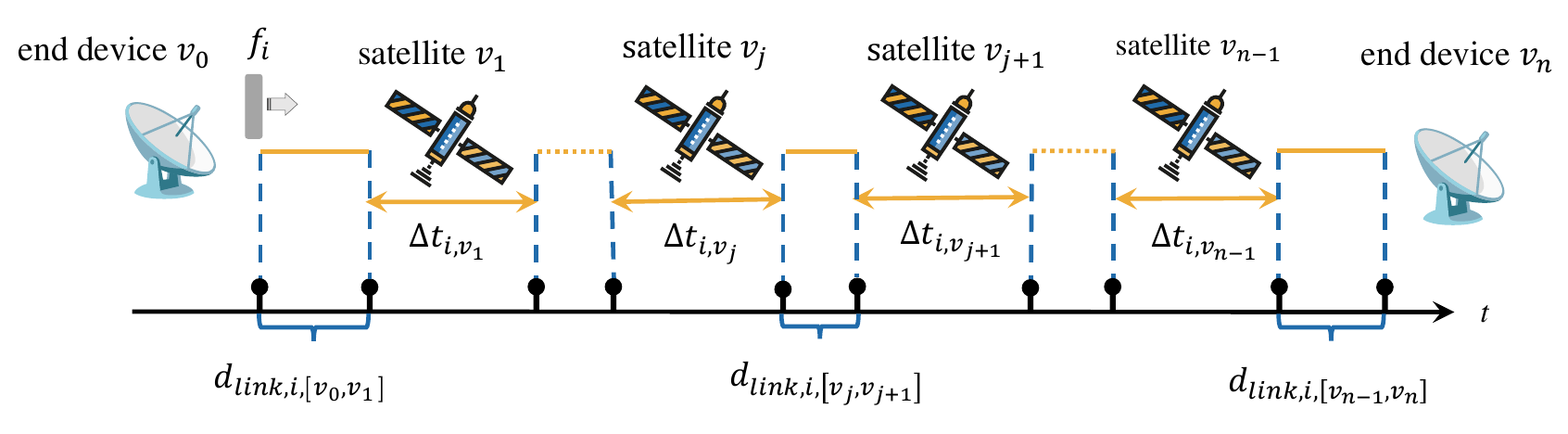}
	\caption{The Residence-Time Mechanism.}
	\label{ReTi_Transmission}
\end{figure*}

Under this mechanism, the e2e delay of flow $f_i$ from sender $v_0$ to receiver $v_n$ is defined as:
\begin{equation}
	Delay_{i}^{e2e} = \sum_{j=0}^{n-1} d_{link, i, [v_j, v_{j+1}]} + \sum_{j=1}^{n-1} \Delta t_{i, v_j}.
\end{equation}

Specifically, when a data frame of flow $f_i$ reaches vertex $v_j$ from $v_{j-1}$, the vertex records its arrival timestamp $Arr_{i,v_j}$ using the local clock and schedules it for forwarding at $Arr_{i,v_j} + \Delta t_{i,v_j}$.
Since this mechanism triggers transmission based on local time rather than the global time synchronization, it functions as a local time-triggered system. Consequently, it remains fully compatible with the guard band strategy of IEEE 802.1Qbv and the preemption strategy of IEEE 802.1Qbu. By utilizing these standards, the switch ensures that TT flows are protected from interference by Non-TT traffic (e.g., Best-Effort), maintaining isolation without global synchronization.

\subsection{Collision Challenges and the CRT Scheduling Problem}

While the residence-time mechanism can stabilize the  e2e delay, it cannot by itself eliminate contention among TT flows. In an asynchronous LEO environment, TT packets may still be blocked by other TT packets on shared links. This is because, unlike the globally synchronized case where the relative transmission offsets between flows remain fixed, different source nodes operate in different local clock domains. When flows from different sources traverse a common link, relative clock drift gradually shifts their transmission windows over time. As a result, even flows that are initially non-overlapping may eventually interfere with each other. This challenge is formalized in the following theorem.

\begin{theorem}[\textbf{Flow Drift}]
	\label{FlowDrift}
	Consider two periodic flows $f_i, f_j \in \mathcal{F}$ originating from different source nodes and sharing a common link $[v_k, v_l]$. Let their initial transmission offsets on this link be $\phi_i$ and $\phi_j$, respectively, and let the relative clock drift rate between their source clock domains be $s>0$. If no explicit synchronization is enforced, then the relative phase difference between the two flows varies over time. Consequently, there exists a finite time after which their transmission windows overlap, causing a collision.
\end{theorem}

\begin{proof}
	Let the initial phase difference be
	\[
	\Delta \phi_0 = |\phi_i - \phi_j|.
	\]
	Due to the relative clock drift $s$, the phase difference evolves over time as
	\[
	\Delta \phi(t) = |\Delta \phi_0 - s t|.
	\]
	Let $C$ denote the transmission duration of a TT frame on the shared link. A collision occurs when the phase difference becomes smaller than the transmission window, i.e.,
	\[
	\Delta \phi(t) < C.
	\]
	Solving this inequality gives
	\[
	t > \frac{\Delta \phi_0 - C}{s}.
	\]
	Therefore, for any non-zero relative clock drift $s>0$, there exists a finite time $t$ such that the transmission windows of $f_i$ and $f_j$ overlap, which causes a collision.
\end{proof}

Collisions among TT flows introduce additional queuing jitter, which can undermine deterministic transmission guarantees. To address this issue, CRT adopts a collision-tolerant scheduling strategy. Since flows generated by the same source are serialized at the ingress, effective collisions only occur among flows originating from different source nodes. We formulate an optimization problem that maximizes schedulability while minimizing the maximum overlap degree, thereby bounding the resulting jitter. Based on this idea, we define the CRT scheduling problem as follows.

\begin{definition}[\textbf{CRT Scheduling Problem}]
	\label{def:crt_problem}
	Given a time-varying network topology sequence $\mathcal{G}$ and a set of time-sensitive flows $\mathcal{F}$, the CRT scheduling problem is to determine a scheduled subset of flows $\mathcal{F}' \subseteq \mathcal{F}$, together with their routing paths $P$ and per-hop residence times $\Delta t$, such that:
	\begin{enumerate}
		\item the number of scheduled flows $|\mathcal{F}'|$ is maximized;
		\item subject to (1), the maximum overlap degree is minimized, where the overlap degree of a link counts only flows originating from distinct source nodes;
		\item all scheduled flows satisfy the e2e deadline, residence-time, and resource constraints.
	\end{enumerate}
\end{definition}

\begin{theorem}
	\label{thm:np_hard}
	The CRT scheduling problem in Definition~\ref{def:crt_problem} is NP-hard.
\end{theorem}

\begin{proof}
	We prove this by reduction from the Edge-Disjoint Paths (EDP) problem~\cite{vygen1995np}, which is NP-complete in directed graphs.
	
	Consider the following decision problem of CRT scheduling: given a topology sequence $\mathcal{G}$, a flow set $\mathcal{F}$, an integer threshold $Q$, and an overlap bound $K$, determine whether there exists a feasible schedule such that at least $Q$ flows are scheduled and the maximum overlap degree is at most $K$.
	
	We reduce an arbitrary EDP instance to a restricted static instance of this decision problem. Given a directed graph $G=(V,E)$ and a set of source-destination pairs $(s_i,d_i)$, we construct a static CRT instance as follows:
	\begin{enumerate}
		\item Topology: Use a single static snapshot identical to $G$.
		\item Flows: For each source-destination pair $(s_i,d_i)$ in the EDP instance, we introduce a private virtual source node $\hat{s}_i$ and connect it to the original source node $s_i$ by a dedicated edge $(\hat{s}_i,s_i)$ that is not shared by any other flow. We then define the corresponding flow $f_i$ to start from $\hat{s}_i$ and terminate at $d_i$.
		
		\item Relaxation: Set link capacities, deadlines, buffer limits, and residence-time bounds sufficiently large so that they never restrict feasibility.
		\item Decision parameters: Set $Q = |\mathcal{F}|$ and $K = 1$.
	\end{enumerate}
	
	Then, there exists a feasible CRT schedule admitting at least $Q$ flows with maximum overlap degree at most $K=1$ if and only if there exists a set of pairwise edge-disjoint paths for all source-destination pairs in the EDP instance. Indeed, admitting at least $Q=|\mathcal{F}|$ flows means all flows must be scheduled. Since the maximum overlap degree is at most $1$, no two scheduled flows can share any link, thus the resulting paths are pairwise edge-disjoint. Conversely, any solution to the EDP instance directly yields a feasible CRT schedule for all flows with maximum overlap degree at most $1$.
	
	Therefore, this restricted decision version of CRT scheduling is NP-hard. Since the optimization problem in Definition~\ref{def:crt_problem} is at least as hard as this restricted version, the CRT scheduling problem is NP-hard.
\end{proof}

\section{Scheduling for CRT Deterministic Transmission}
\subsection{Scheduling Model}

\label{Problem Formulation}

\begin{figure}[htbp]
	\centering
	\includegraphics[width=0.38\textwidth]{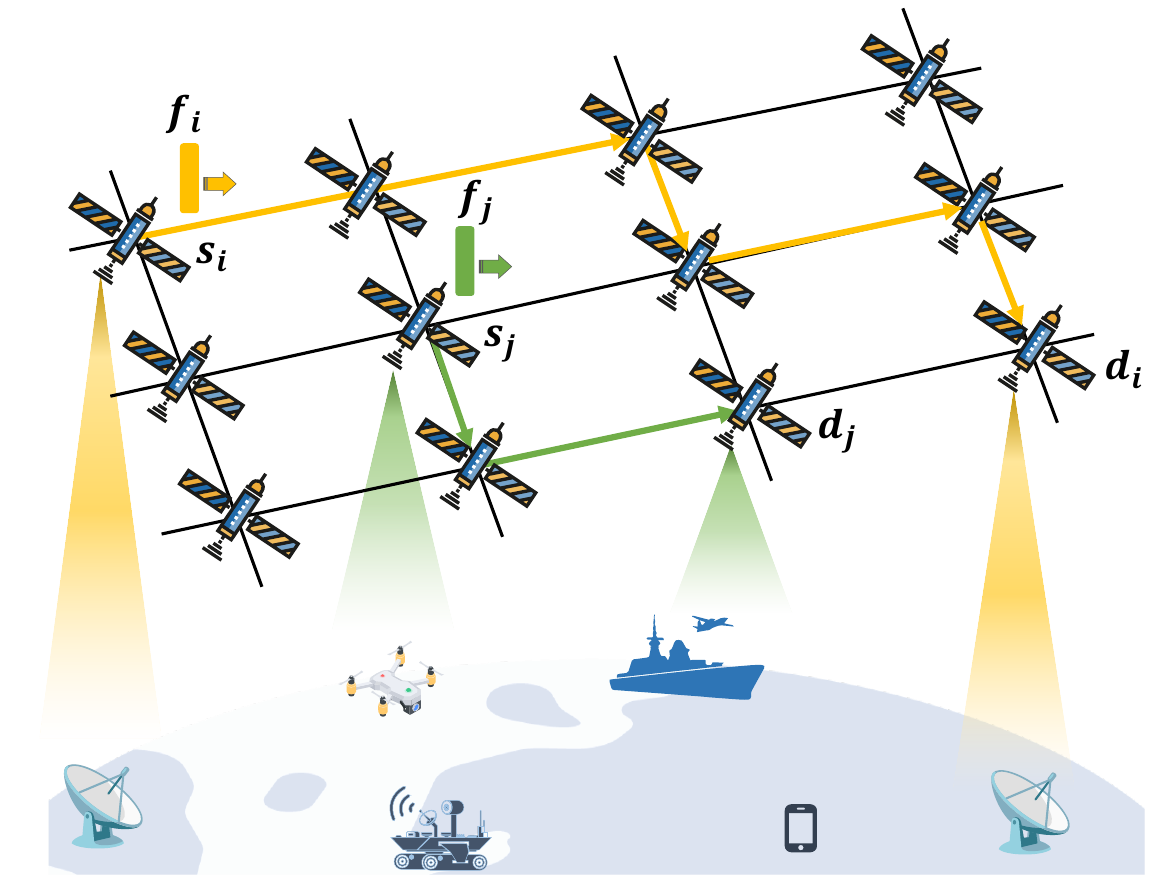}
	\caption{TT scheduling model in LEO satellite network.}
	\label{network model}
\end{figure}

As depicted in Fig. \ref{network model}, the network architecture comprises access devices and a LEO satellite system. The terminal devices handle the sending and receiving of each TT flow, while the LEO satellite network manages the scheduling process. A summary of the important notations and definitions used throughout the paper is provided in Table \ref{tab:notation}, aiding in the comprehension of the proposed models and algorithms.

\begin{table}[htbp]
	\centering
	\caption{Notations and Definitions}
	\begin{tabular}{>{\centering\arraybackslash}p{1.9cm}p{6cm}}
		\toprule
		\textbf{Notation} & \textbf{Description} \\ 
		$\mathcal{T}$ & The set of discrete time slots \\ 
		$\tau$ & The index of a time slot \\
		$G^\tau$ & The network topology graph in time slot $\tau$ \\
		$\mathcal{L}^\tau$ & The set of directed ISLs in time slot $\tau$ \\
		$B_e$ & The bandwidth capacity of link $e$ \\ 
		$d_{prop,e}^\tau$ & The propagation delay of link $e$ in time slot $\tau$ \\ 
		$d_{proc, v}$ & The fixed processing delay at node $v$ \\ 
		
		$\mathcal{F}$ & The set of TT flows \\ 
		$f_i$ & The $i$-th TT flow \\ 
		$T_i$ & The transmission period of $f_i$ \\ 
		$L_i$ & The frame length of $f_i$ \\ 
		$D_i$ & The e2e deadline of $f_i$ \\ 
		$C_{i,e}$ & The transmission time of $f_i$ on link $e$ ($L_i/B_e$) \\
		$C_{max,e}$ & The maximum transmission time of a TT frame on link $e$ \\
		
		$p_i^\tau$ & The routing path of flow $f_i$ in time slot $\tau$ \\ 
		$n_e^\tau$ & The overlap degree on link $e$ \\ 
		$WCD_e^\tau(n_e^\tau)$ & The worst-case delay caused by collisions on link $e$ \\
		$\Delta t_{i,v}^\tau$ & The allocated residence time of $f_i$ at node $v$ \\ 
		$t_{extra}$ & The dynamic regulatory time to absorb jitter \\
		$D_i^{target}$ & The scheduled deterministic e2e latency \\ 
		
		$y_i$ & Binary schedule variable ($1$ if $f_i$ is scheduled) \\ 
		$x_{i,p}^\tau$ & Binary path variable ($1$ if $f_i$ uses path $p$) \\ 
		$\delta_{e,p}$ & Binary indicator ($1$ if link $e$ is on path $p$) \\
		\bottomrule
	\end{tabular}
	\label{tab:notation}
\end{table}

We consider a LEO satellite constellation consisting of $M$ orbital planes with $N$ satellites per plane. The network connectivity relies on ISLs. Due to the periodic nature of satellite orbits, the dynamic topology is discretized into a sequence of time slots $\mathcal{T}=\{1,\dots,m\}$. Within each time slot $\tau$, the topology is modeled as a static directed graph $G^\tau=(V, \mathcal{L}^\tau)$, where $V$ denotes the set of satellites and $\mathcal{L}^\tau$ denotes the set of ISLs.

The network transmits TT flows requiring strictly deterministic service. We formally define a TT flow $f_i \in \mathcal{F}$ as a tuple:
\begin{equation}
	f_i = \{ T_i, L_i, s_i, d_i, D_i \},
\end{equation}
where $T_i$ is the transmission period, $L_i$ is the frame length, $D_i$ is the maximum allowable e2e deadline. $s_i$ and $d_i$ represent the source and destination satellites corresponding to the sender and receiver terminal devices. The packet transmission time on a link $e$ with bandwidth $B_e$ is denoted as $C_{i,e} = L_i / B_e$. 

Given the time-varying topology, the routing path of flow $f_i$ in time slot $\tau$ is denoted as $p_i^\tau$. It consists of a sequence of links:
\begin{equation}
	p_i^\tau = \left[ [v_0, v_1], [v_1, v_2], \ldots, [v_{n-1}, v_{n}] \right].
\end{equation}

The physical delay on an ISL consists of transmission delay and propagation delay. For flow $f_i$ on link $e$, the link delay is:
\begin{equation}
	d_{link,i,e}^\tau = d_{prop,e}^\tau + C_{i,e}.
\end{equation}

The e2e latency assigned by the residence-time mechanism is:
\begin{equation}
	\label{det}
	\sum_{e \in p_i^\tau} d_{link,i,e}^\tau +\sum_{v \in p_i^\tau} \Delta t_{i,v}^\tau = D_i^{target},\quad \forall f_i\in\mathcal{F},\ \forall \tau\in\mathcal{T}.
\end{equation}

This equality defines the e2e delay $D_i^{target}$
before accounting for collision-induced jitter.

The residence time $\Delta t_{i,v}^\tau$ at node $v$ is the control variable:
\begin{equation}
	\Delta t_{i,v}^\tau = d_{proc, v} + t_{extra,i,v}^\tau ,
	\label{eq:residence_def}
\end{equation}
where $d_{proc, v}$ is the processing delay of node $v$,
it mainly results from analyzing the frame header and performing frame verification, which is considered a fixed constant.
$t_{extra,i,v}^\tau$ is the additional residence time required by the scheduler at node $v$. This residence time is dynamically adjusted by the scheduler to stabilize the e2e delay and compensate for variations in link delays across time slots.

\subsection{Collision-Induced Jitter Analysis}
\label{subsec:wcd_analysis}

In a network-wide traffic scenario without reliance on global time synchronization, two main challenges arise: conflicts between TT and non-TT traffic, and conflicts among TT traffic flows themselves. 
The first issue is addressed through the residence-time mechanism combined with standard TSN features. Specifically, TT packets avoid blocking from non-TT packets via IEEE 802.1Qbu frame preemption or IEEE 802.1Qbv guard-band mechanisms. 

However, blocking from other TT packets persists. In an asynchronous environment, data flows originating from different sources operate in independent time domains. As these flows converge on a shared inter-satellite link $e$, their relative arrival times drift. Consequently, packets from different flows may arrive simultaneously, leading to resource contention.
We define the overlap degree $n_e^\tau$ as the number of scheduled flows
originating from distinct source nodes and traversing link $e$ in time slot $\tau$.

\begin{equation}
	\label{eq:overlap_degree}
	n_e^\tau=
	\left|
	\left\{
	s_i\;\middle|\; f_i\in \mathcal{F},\ e\in p_i^\tau,\ f_i\ \text{is scheduled in slot}\ \tau
	\right\}
	\right|.
\end{equation}

\subsubsection{Overlap and Network Calculus Modeling}
To rigorously quantify the physical jitter caused by these collisions, we must identify the effective interference. 
Crucially, TT flows originating from the same source node are inherently serialized at the ingress interface and do not collide with each other. Collisions only occur between flows from different sources.
For a target flow $f_i$ on link $e$, the set of effective competing flows is denoted as $\mathcal{O}_{i,e}^\tau = \{ f_j \mid e \in p_j^\tau, s_j \neq s_i \}$. The number of effective interfering flows is $k_{i,e}^\tau = |\mathcal{O}_{i,e}^\tau|$.

We employ network calculus to derive the WCD. In the worst-case asynchronous scenario, packets from all flows in $\mathcal{O}_{i,e}^\tau$ may arrive simultaneously at the switch egress port.
The aggregate arrival curve $\alpha(t)$ for these interfering flows is bounded by their collective burst size:
\begin{equation}
	\alpha(t) = \sum_{f_j \in \mathcal{O}_{i,e}^\tau} L_j + \rho_{agg} \cdot t,
\end{equation}
where $L_j$ is the packet length and $\rho_{agg} = \sum_{f_j \in \mathcal{O}_{i,e}^\tau} (L_j/T_j)$ represents the aggregate average arrival rate. Since TT flows are served strictly by the physical link bandwidth $B_e$, the service curve $\beta(t)$ is linear:
\begin{equation}
	\beta(t) = B_e \cdot t.
\end{equation}

\subsubsection{Jitter Bound Derivation}
The delay bound is defined as the maximum horizontal deviation between $\alpha(t)$ and $\beta(t)$. For a packet of flow $f_i$, the WCD occurs when it is processed last among the contending burst. We define the link-level worst-case jitter on link $e$ in time slot $\tau$, denoted as $WCD_e^\tau(n_e^\tau)$, as:
\begin{equation}
	\label{eq:wcd_link_level}
	WCD_e^\tau(n_e^\tau) = \frac{\sum_{f_j \in \mathcal{O}_{i,e}^\tau} L_j}{B_e} = \sum_{f_j \in \mathcal{O}_{i,e}^\tau} C_{j,e},
\end{equation}
where $C_{j,e} = L_j / B_e$ is the transmission time. To ensure a strictly deterministic upper bound independent of specific flow lengths in the set, we simplify this as:
\begin{equation}
	\label{eq:wcd_formula_final}
	WCD_e^\tau(n_e^\tau) \le (n_e^\tau - 1) \cdot C_{max,e},
\end{equation}
where $n_e^\tau$ is the overlap degree, i.e., the number of distinct-source
flows on link $e$, and $C_{max,e}$ is the maximum possible transmission
time of any TT frame on link $e$.

This $WCD_e^\tau(n_e^\tau)$ represents the unavoidable collision-induced jitter contributed by by link $e$ in slot $\tau$. Although this jitter is not part
of the target e2e delay enforced by the residence-time mechanism, it consumes
the available deadline margin. Therefore, the cumulative jitter along routing
path $p_i^\tau$ must satisfy
$\sum_{e \in p_i^\tau} WCD_e^\tau(n_e^\tau) \le D_i - D_i^{target}$.
Consequently, minimizing the overlap degree $n_e^\tau$ is critical for deterministic transmission.

\subsection{Optimization Framework}
The optimization goal is to maximize the schedulability while minimizing the overlap degree, so as to suppress collision-induced jitter.

We define two sets of binary decision variables:

\begin{equation}
	x_{i,p}^\tau = 
	\begin{cases} 
		1, & \text{if flow } f_i \text{ uses path } p \in \mathcal{P}_i^\tau \text{ in slot } \tau, \\
		0, & \text{otherwise}.
	\end{cases}
\end{equation}

\begin{equation}
	y_i = 
	\begin{cases} 
		1, & \text{if flow } f_i \text{ is  scheduled}, \\
		0, & \text{if flow } f_i \text{ is unscheduled}.
	\end{cases}
\end{equation}

Let $\delta_{e,p}$ be a binary parameter, where $\delta_{e,p}=1$ if link $e$
belongs to path $p$, and 0 otherwise. Let $\mathcal{V}_{src}$ denote the set
of source nodes associated with the flow set $\mathcal{F}$. To express the
source-level overlap degree in terms of the path selection variables, we
introduce an auxiliary binary variable
\begin{equation}
	\label{eq:source_indicator}
	u_{s,e}^\tau =
	\begin{cases}
		1, & \text{if source } s \text{ occupies link } e \text{ in slot } \tau,\\
		0, & \text{otherwise,}
	\end{cases}
\end{equation}
for all $s \in \mathcal{V}_{src}$, $e \in \mathcal{L}^\tau$, and
$\tau \in \mathcal{T}$.
Here, source $s$ occupies link $e$ if at least one scheduled flow
originating from $s$ traverses $e$ in slot $\tau$.

Then, the overlap degree on link $e$ in slot $\tau$ can be written as
\begin{equation}
	\label{eq:source_overlap_ilp}
	n_e^\tau = \sum_{s \in \mathcal{V}_{src}} u_{s,e}^\tau,
	\quad \forall e \in \mathcal{L}^\tau,\ \forall \tau \in \mathcal{T}.
\end{equation}

To link $u_{s,e}^\tau$ with the path selection variables, we impose
\begin{equation}
	\label{eq:source_overlap_link}
	\begin{aligned}
		u_{s_i,e}^\tau &\ge x_{i,p}^\tau \cdot \delta_{e,p},\\
		&\forall f_i \in \mathcal{F},\ \forall p \in \mathcal{P}_i^\tau,\ 
		\forall e \in \mathcal{L}^\tau,\ \forall \tau \in \mathcal{T}.
	\end{aligned}
\end{equation}
\subsubsection{Objective Function}
\label{subsec:objective}

we formulate a hierarchical optimization problem with priority order $P_1 \succ P_2$.

\paragraph{Priority 1 ($P_1$): Maximize Schedulability}
Our primary goal is to maximize the total number of scheduled TT flows. The objective function is:

\begin{equation}
	\label{eq:obj_P1}
	\max \ \mathcal{J}_1 \triangleq \sum_{f_i \in \mathcal{F}} y_i.
\end{equation}

\paragraph{Priority 2 ($P_2$): Minimize Collision Intensity}
Subject to maintaining the optimal schedulability $\mathcal{J}_1^\star$, denoting the optimal value of Eq. \eqref{eq:obj_P1}, we minimize the maximum overlap degree across all links. 
To linearize the Min-Max objective ($\min \max n_e^\tau$), we introduce an auxiliary integer variable $z$. The second-stage objective is
\begin{equation}
	\label{eq:obj_P2}
	\min \ \mathcal{J}_2 \triangleq z,
\end{equation}
subject to
\begin{equation}
	\label{eq:z_def}
	n_e^\tau \le z, \quad \forall \tau \in \mathcal{T},\ \forall e \in \mathcal{L}^\tau,
\end{equation}
and
\begin{equation}
	\label{eq:fix_P1_opt}
	\sum_{f_i \in \mathcal{F}} y_i = \mathcal{J}_1^\star .
\end{equation}

In summary, the overall objective can be expressed as a lexicographic optimization:
\begin{equation}
	\label{eq:lex_obj}
	\text{lex} \ \max_{x,y} \big[ \mathcal{J}_1(y),\ -\mathcal{J}_2(x) \big].
\end{equation}

\subsubsection{Maximum Delay Constraint}
Since the target e2e delay $D_i^{target}$ serves as a baseline and does not
account for asynchronous collisions, a safety margin must be reserved for
collision-induced jitter. Therefore, the sum of the baseline delay and the
cumulative WCD along the path must not exceed the deadline:
\begin{equation}
	\label{eq:max_delay_constraint}
	D_i^{target} + \sum_{e \in p_i^\tau} WCD_e^\tau(n_e^\tau)
	\le D_i,\quad \forall f_i \in \mathcal{F},\ \forall \tau \in \mathcal{T}.
\end{equation}
Here, $\sum_{e \in p_i^\tau} WCD_e^\tau(n_e^\tau)$ represents the maximum collision-induced jitter caused by link overlaps along the path.

\subsubsection{Baseline Latency Consistency Constraint}

The residence-time mechanism assigns a baseline e2e latency $D_i^{target}$ for each scheduled flow. Eq.~(\ref{det}) enforces this baseline delay, while the actual packet delay may additionally include bounded collision-induced jitter.

\subsubsection{Path Uniqueness Constraint} To ensure routing consistency, each scheduled flow is assigned exactly one path per time slot:
\begin{equation}
	\sum_{p \in \mathcal{P}_i^\tau} x_{i,p}^\tau = y_i, \quad \forall f_i \in \mathcal{F}, \forall \tau \in \mathcal{T}.
\end{equation}

\subsubsection{Bandwidth Constraint} The aggregated bandwidth of all scheduled TT flows traversing a specific link must not exceed the link's bandwidth capacity. 
\begin{equation}
	\sum_{i \in \mathcal{F}} \sum_{p \in \mathcal{P}_i^\tau} x_{i,p}^\tau \cdot \delta_{e,p} \cdot \frac{L_i}{T_i} \le B_e, \quad \forall e \in \mathcal{L}^\tau, \forall \tau \in \mathcal{T}.
\end{equation}

\subsubsection{Buffer Capacity Constraint}
Satellites employ a store-and-forward strategy, where insufficient buffer capacity in any time slot $\tau$ inevitably leads to packet loss. Although dedicated
buffer spaces can be allocated for TT traffic to isolate it from
non-TT traffic interference, contention persists among concurrent
TT flows occupying the shared TT-specific buffer space within
each time slot. To prevent buffer overflow, 
we impose an upper bound on the residence time at each hop:
\begin{equation}
	\label{eq:buffer_limit}
	\Delta t_{i,v}^\tau \le T_{buffer}^{max}, \quad \forall f_i \in \mathcal{F}, \forall v \in p_i^\tau,
\end{equation}
where $T_{buffer}^{max}$ represents the maximum duration a packet can be buffered at a switch port.

\subsubsection{Minimum Residence Time Constraint}
Physically, a packet cannot be processed faster than the hardware limit. Therefore, the allocated residence time must be lower-bounded by the processing delay, regardless of whether collisions occur,
\begin{equation}    
	\Delta t_{i, v_j}^\tau \geq d_{proc}, \quad \forall f_i \in \mathcal{F},  \forall v_j \in p_i^\tau.
\end{equation}


\section{Algorithm Design}
\label{sec:algorithm}

\subsection{Heuristic Strategy: Iterative Layering}
Since the formulated CRT scheduling problem is NP-Hard, we propose a heuristic algorithm CRT-Fast.

Unlike simple greedy approaches that schedule flows strictly one-by-one, CRT-Fast adopts a layer-by-layer superposition strategy. The core idea is to decompose the complex global optimization into a sequence of simpler sub-problems. In each layer, we solve a constrained version of the "Maximum Independent Set" problem: finding the maximum subset of remaining flows that can be added to the network without conflicting with each other and without violating the deadlines of already scheduled flows.

This strategy naturally controls the collision intensity:
\begin{itemize}
	\item Layer 1: We select the maximum set of edge-disjoint flows. These form the base layer ($n_e=1, WCD=0$).
	\item Layer $k$: We select the next batch of edge-disjoint flows and superimpose them onto the existing schedule. This increases the overlap degree $n_e$ by at most 1 per layer, ensuring the interference grows gradually and predictably.
	\item Temporal Consistency: When scheduling flow $f_i$ in slot $\tau$, we prioritize reusing the path $p_i^{\tau-1}$ from the previous slot. If $p_i^{\tau-1}$ remains valid in the current topology and satisfies constraints, we select it immediately. This avoids unnecessary path switching during topology handovers.
\end{itemize}

\subsection{Handling Dynamic Topology: Seamless Handover}
\label{subsec:handover}

While the discretized time-slot model ($G^\tau$) effectively handles intra-slot stability, LEO networks face critical challenges during the transition between time slots. As satellites move, ISLs may break, causing packet loss if the flow path is not updated instantaneously.
To mitigate this, we propose a seamless handover strategy based on path lookahead, which acts as a guiding principle for our scheduling algorithm.

\subsubsection{Implicit Backup Path}

\begin{figure*}[htbp]
	\centering
	\includegraphics[width=0.78\textwidth]{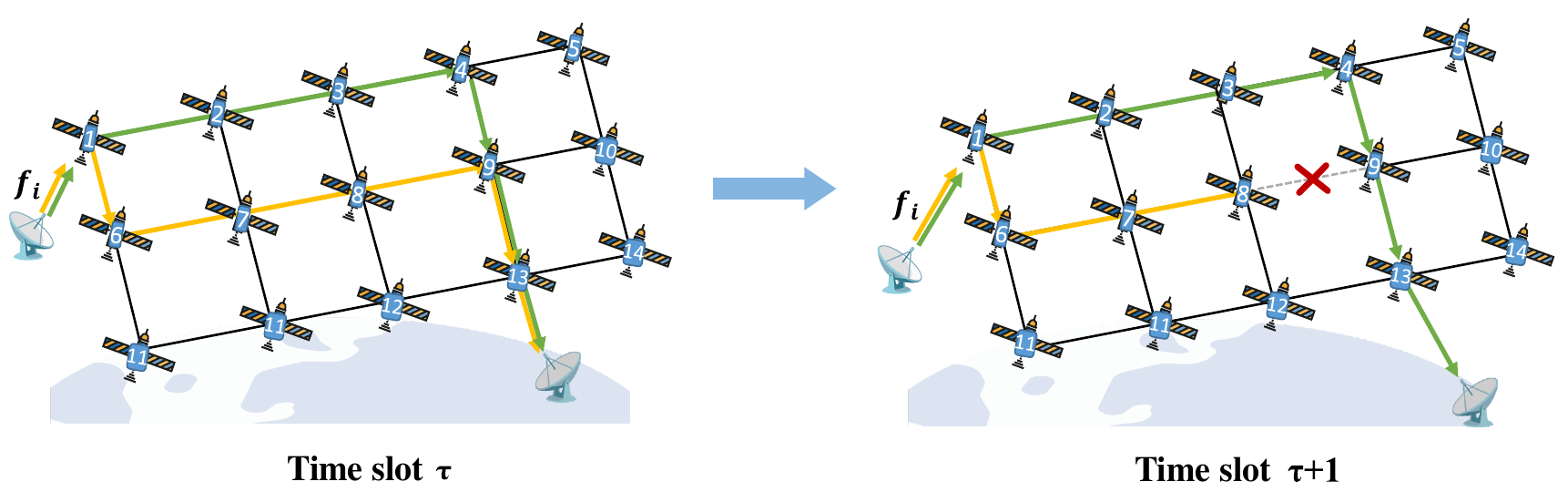}
	\caption{Example of Cross-Time-Slot Transmission.}
	\label{scheduling_example}
\end{figure*}

Unlike traditional fault tolerance schemes that reserve dedicated redundant resources, which would increase collision intensity $n_e^\tau$ and worsen jitter, we adopt an efficient implicit backup approach.
As illustrated in Fig. \ref{scheduling_example}, for a flow $f_i$ in the current time slot $\tau$:
The primary path is the optimal path selected for slot $\tau$. The backup path is essentially the pre-calculated primary path for the next time slot $\tau+1$.
Since the backup path is intended for the future, it does not actively consume bandwidth or contribute to the overlap degree ($n_e^\tau$) in the current slot, thus preserving the low-jitter performance of our collision-tolerant mechanism.

\subsubsection{Switching Mechanism and Path Continuity}
During the transition phase near the end of slot $\tau$, if the primary link becomes unstable or disconnects, the switch immediately redirects packets to the backup path. This ensures continuous data flow.
To support this mechanism efficiently, our scheduling strategy enforces a principle of path continuity. When scheduling for a new time slot $\tau$, the system prioritizes keeping the path unchanged from the previous slot ($p^\tau = p^{\tau-1}$) whenever the topology allows. This reduces the frequency of switching paths and the risk of transient congestion caused by path changes.

\subsection{Algorithm Description}

Main Framework (Alg.~\ref{alg:crt}):
The main algorithm manages the lifecycle of the scheduling process. It maintains a set of unscheduled flows and the global network state (overlap counters $n_e^\tau$).
The process consists of two phases:
\begin{enumerate}
	\item Iterative Superposition: In each iteration, Alg. \ref{alg:crt} calls the subroutine Alg. \ref{alg:findset} to identify the best set of flows for the current layer. If a non-empty set is returned, these flows are committed to the schedule, and the per-link source sets together with the corresponding overlap counters are updated accordingly. This loop continues until no more flows can be scheduled.
	\item Deterministic Parameter Calculation: Once the slot-level paths of a flow are finalized, the post-processing stage computes a common target delay and the corresponding per-hop residence times over its active time slots. Specifically, algorithm first evaluates the path delay in each slot and sets the common baseline target delay to the maximum among them. It then computes the slot-dependent slack as the difference between this common baseline and the current slot's delay, and evenly distributes the slack across all hops on the selected path to obtain the final residence times. CheckGlobalFeasibility validates local path and resource feasibility during path selection, while the exact cross-slot residence-time allocation is finalized in this post-processing stage.
	
\end{enumerate}

\begin{algorithm}[!t]
	\small
	\DontPrintSemicolon
	\SetAlgoLined
	\caption{CRT-Fast algorithm}
	\label{alg:crt}
	\KwIn{$G^\tau$, Flow set $\mathcal{F}$, Candidate path sets $\{\mathcal{P}_i^\tau\}$}
	\KwOut{Scheduled Paths $\mathbb{P}$, Residence Times $\mathbf{\Delta t}$}
	
	\tcc{1. Initialization}
	Initialize $\mathbb{P} \leftarrow \emptyset$, Unscheduled Set $\mathcal{U} \leftarrow \mathcal{F}$\;
	Initialize Global Link Counters $n_e^\tau \leftarrow 0$ for all $e \in \mathcal{L}^\tau$\;
	Initialize $SrcOnLink[e]\gets\emptyset,\forall e\in L^\tau$\;
	
	\tcc{2. Layer-by-Layer Scheduling}
	\While{$\mathcal{U} \neq \emptyset$}{
		$\mathcal{S}_{layer} \leftarrow$ \textbf{FindMaxFeasibleIS}($\mathcal{U},\mathbb{P},n_e^\tau,SrcOnLink,\mathbb{P}^{\tau-1}$)\;
		
		\If{$\mathcal{S}_{layer} == \emptyset$}{
			\textbf{break} \tcp*{Network saturated, stop iteration}
		}
		
		\tcc{Commit Layer and Update Global State}
		\ForEach{$(f_i, p) \in \mathcal{S}_{layer}$}{
			$\mathbb{P}[f_i] \leftarrow p$\;
			Remove $f_i$ from $\mathcal{U}$\;
			\ForEach{$e\in p$}{
				\If{$s_i \notin SrcOnLink[e]$}{
					$SrcOnLink[e]\gets SrcOnLink[e]\cup\{s_i\}$\;
					$n_e^\tau[e]\gets n_e^\tau[e]+1$\;
				}
			}
		}
	}
	
	\tcc{3. Post-Processing: Cross-Slot Residence-Time Allocation}
	\ForEach{$f_i \in \mathbb{P}$}{
		\tcp{Compute the path delay in each slot}
		\ForEach{$\tau \in \mathcal{T}_i$}{
			$p \leftarrow \mathbb{P}[f_i,\tau]$\;
			$D_{fixed,i}^\tau \leftarrow \sum_{e \in p} d_{link,i,e}^\tau + \sum_{v \in p} d_{proc,v}$\;
			$WCD_{total,i}^\tau \leftarrow \sum_{e \in p} (n_e^\tau-1)C_{max,e}$\;
		}
		
		\tcp{Common baseline delay over all slots}
		$D_i^{target} \leftarrow \max_{\tau \in \mathcal{T}_i} D_{fixed,i}^\tau$\;
		
		\tcp{Allocate slot-dependent slack to each hop}
		\ForEach{$\tau \in \mathcal{T}_i$}{
			$p \leftarrow \mathbb{P}[f_i,\tau]$\;
			$Slack_i^\tau \leftarrow D_i^{target} - D_{fixed,i}^\tau$\;
			
			\If{$D_i^{target} + WCD_{total,i}^\tau > D_i$}{
				mark $f_i$ infeasible\;
				\textbf{break}
			}
			
			\ForEach{$v \in p$}{
				$t_{extra,i,v}^\tau \leftarrow \frac{Slack_i^\tau}{|p|}$\tcp*{$|p|=|\{v\,|\,v\in p_i^\tau\}|$}
				$\Delta t_{i,v}^\tau \leftarrow d_{proc,v} + t_{extra,i,v}^\tau$\;
			}
		}
	}
	\Return $\mathbb{P}, \mathbf{\Delta t}$\;
\end{algorithm}

Finding Feasible Independent Set (Alg.~\ref{alg:findset}):
This subroutine constructs the current layer $\mathcal{S}_{layer}$ using a minimum-conflict greedy rule to mitigate collision-induced jitter.

\begin{algorithm}[!t]
	\small
	\DontPrintSemicolon
	\SetAlgoLined
	\caption{Subroutine: FindMaxFeasibleIS}
	\label{alg:findset}
	\KwIn{Unscheduled $\mathcal{U}$, Scheduled $\mathbb{P}$, Global Counters $n_e^\tau$, Per-Link Source Sets $SrcOnLink$, Previous Paths $\mathbb{P}^{\tau-1}$}
	\KwOut{Layer Subset $\mathcal{S}_{layer}$}
	\BlankLine
	
	$\mathcal S_{layer},LL\gets\emptyset$\;
	Initialize $LinkCounts[e]\gets 0,\forall e\in G^\tau$\;
	Initialize $TmpSrc[e] \gets \emptyset,\ \forall e \in G^\tau$\;
	
	\tcc{1. Conflict Degree Calculation}
	\ForEach{$f_i \in \mathcal U$}{
		\ForEach{$e \in \bigcup_{p \in \mathcal P_i^\tau[1:K]} p$}{
			\If{$s_i \notin TmpSrc[e]$}{
				$TmpSrc[e] \gets TmpSrc[e] \cup \{s_i\}$\;
				$LinkCounts[e] \gets LinkCounts[e] + 1$\;
			}
		}
	}
	\ForEach{$f_i\in\mathcal U$}{
		$(cd_i,p_i^\star)\gets \min_{p\in \mathcal P_i^\tau[1{:}K]} \sum_{e\in p}(LinkCounts[e]-1)$\;
	}
	Sort $\mathcal U$ by $cd_i\uparrow$ ($D_i$ tie)\;
	
	\tcc{2. Greedy Selection}
	\ForEach{$f_i\in\mathcal U$}{
		$p_{best}\gets\emptyset$; $p\gets \mathbb P^{\tau-1}[f_i]$\;
		\If{$p\neq\emptyset$ \textbf{and} IsPathValid($p,G^\tau$) \textbf{and} $p\cap LL=\emptyset$
			\textbf{and} CheckGlobalFeasibility($f_i,p,\mathbb P\cup\mathcal S_{layer},n_e^\tau,SrcOnLink$)}{
			$p_{best}\gets p$\;
		}
		\If{$p_{best}=\emptyset$}{
			$p\gets p_i^\star$\;
			\If{$p\neq\emptyset$ \textbf{and} $p\cap LL=\emptyset$ \textbf{and}
				CheckGlobalFeasibility($f_i,p,\mathbb P\cup\mathcal S_{layer},n_e^\tau,SrcOnLink$)}{
				$p_{best}\gets p$\;
			}
		}
		\If{$p_{best}=\emptyset$}{
			\ForEach{$p\in \mathcal P_i^\tau[1{:}K]$}{
				\If{$p=p_i^\star$ \textbf{or} $p\cap LL\neq\emptyset$ \textbf{or}
					CheckGlobalFeasibility($f_i,p,\mathbb P\cup\mathcal S_{layer},n_e^\tau,SrcOnLink$)==False}{
					\textbf{continue}\;
				}
				$p_{best}\gets p$; \textbf{break}\;
			}
		}
		\If{$p_{best}\neq\emptyset$}{
			Add $(f_i,p_{best})$ to $\mathcal S_{layer}$\;
			$LL\gets LL\cup\{e\mid e\in p_{best}\}$\;
		}
	}
	\Return $\mathcal{S}_{layer}$\;
\end{algorithm}

\begin{enumerate}
	\item Conflict-degree ordering.
	For each unscheduled flow $f_i$, we estimate link popularity via $LinkCounts[e]$, i.e., the number of unscheduled flows whose candidate paths traverse link $e$.
	For any candidate path $p$, its conflict score is $pc(p)=\sum_{e\in p}(LinkCounts[e]-1)$.
	We define the flow conflict degree as $\mathrm{cd}(f_i)=\min_{p\in\mathcal{P}_i^\tau} pc(p)$ and record $p_i^\star=\arg\min pc(p)$.
	Flows are processed in ascending $\mathrm{cd}(\cdot)$ (deadline $D_i$ breaks ties). For consistency with the overlap definition, the link popularity used in conflict-degree estimation is accumulated at the source level, i.e., multiple flows from the same source contribute only once to a given link.
	
	\item Path continuity.
	To support seamless handover (Sec.~\ref{subsec:handover}), we first reuse the previous-slot path $p_{prev}$ if it remains valid, layer-conflict-free, and globally feasible.
	
	\item Low-conflict path selection.
	If $p_{prev}$ is not eligible, we try $p_i^\star$ first; otherwise we scan the remaining candidates and select the first path that satisfies (i) intra-layer edge-disjointness and (ii) CheckGlobalFeasibility under updated $n_e^\tau$. The procedure CheckGlobalFeasibility tentatively inserts a candidate path $p$ of flow $f_i$ into the current partial schedule $\mathbb{P}\cup \mathcal{S}_{layer}$ and updates the overlap counters according to distinct source nodes. It then checks whether the candidate flow and all already scheduled flows sharing at least one link with $p$ remain feasible. In particular, it recomputes the cumulative WCD, derives the corresponding target delay $D_i^{target}=D_i-WCD_{total}$, and verifies that the resulting slack is non-negative and that the per-hop residence times satisfy the buffer bound $T_{buffer}^{max}$
	
	\item Overlap accounting.
	When computing WCD, only interference from distinct sources is counted due to source serialization.
\end{enumerate}

\subsection{Complexity Analysis}
The computational complexity of CRT-Fast is determined by the iterative layering process executed across $|\mathcal{T}|$ discrete time slots. In the worst-case scenario characterized by dense conflicts, each layer admits only a single flow, resulting in $|\mathcal{F}|$ iterations per slot. Within each iteration, computing conflict degrees requires aggregating link popularity over the $K$ candidate paths and evaluating the minimum path-conflict score for each flow, which costs
$O(|\mathcal{F}| \cdot K \cdot L_{path} + |\mathcal{F}| \log |\mathcal{F}|)$.
Subsequently, the path selection evaluates up to $K$ candidates per flow, and the bottleneck remains CheckGlobalFeasibility, which may validate against up to $O(|\mathcal{F}|)$ previously scheduled flows along a path of length $L_{path}$.
Therefore, the overall worst-case time complexity is bounded by $O(|\mathcal{T}| \cdot |\mathcal{F}|^2 \cdot K \cdot L_{path})$.
Although the path continuity strategy significantly reduces the average execution time by avoiding redundant searches in stable topologies, this polynomial worst-case bound guarantees that the algorithm remains scalable for large-scale LEO satellite networks with thousands of flows.

\section{Performance Evaluation}
\label{sec:evaluation}

In this section, we evaluate the performance of the proposed CRT-Fast algorithm on a simulated LEO satellite network.

\subsection{Experimental Setup}


\textbf{Simulation Environment:} The simulations were conducted using the Iridium and the Starlink constellation. The Iridium constellation is a representative polar-orbit satellite network consisting of 6 orbit planes with 11 satellites per plane, while the Starlink constellation consists of 1,584 satellites in 72 orbits. The dynamic network topology is discretized into a sequence of $m=10$ time slots, with each slot duration set to 10 seconds.
To reflect a realistic multi-service network scenario, we assume the deterministic traffic operates within a dedicated network slice. Accordingly, the ISLs are configured with a bandwidth of 100 Mbps reserved specifically for TT flows. 

\textbf{Traffic Generation and Candidate Paths:}
TT flows are generated by randomly selecting source-destination pairs from the satellite nodes. Each flow has a fixed frame size of 1,500 Bytes and a fixed period of 10 ms. The deadline is defined as
\[
D_i=\min(\alpha d_i^{phy},\ d_i^{phy}+\Delta_{buf},\ D_{max}),
\]
where $d_i^{phy}$ is the shortest-path propagation delay. We set $(\alpha, \Delta_{buf}, D_{max})=(1.5,   30\text{ ms}, 100\text{ ms})$ for Iridium and $(2.0, 80\text{ ms}, 500\text{ ms})$ for Starlink. Candidate paths are generated using Yen’s K-shortest simple paths algorithm~\cite{yen1971finding}, with $K=5$.

\textbf{Baseline Algorithms:}
To validate the effectiveness of CRT-Fast, we compare it against four representative baselines:
\begin{itemize}
	\item DSTMR \cite{jiang2023spatio}: It combines spatio-temporal routing with conservative resource reservation under the same parameter setting as CRT-Fast.  
	\item Strict Non-Overlapping: A conservative lower-bound baseline for schedulability. It enforces a strict non-overlapping constraint ($n_e^\tau \le 1$) and rejects any flow that cannot find a conflict-free path.
	\item Shortest Path First (SPF): A static-routing baseline. It selects the single shortest path for each flow.
	\item Load-Aware Greedy (LAG): A representative dynamic-routing greedy baseline. It sequentially schedules flows by selecting the path with the minimum current congestion cost (Best-Fit).
\end{itemize}

All algorithms are evaluated on the same topology snapshots and TT flow sets. Each experiment is repeated over 5 independent runs with different random seeds, and all reported results are averaged across runs.

\subsection{Performance Analysis}

\textbf{Schedulability Evaluation:}
Fig.~\ref{fig:performance} shows the average scheduling success rate as the traffic load increases on Iridium. CRT-Fast consistently achieves the highest schedulability, maintaining 100\% admission at light loads and 94.5\% under the heaviest tested load. In contrast, the Strict Non-Overlapping drops rapidly as the load increases, because fully conflict-free routing severely limits the feasible solution space. SPF performs reasonably well at light loads but degrades more quickly at higher loads due to its rigid shortest-path selection. LAG improves over SPF by using dynamic routing to alleviate local congestion, but it still performs worse than CRT-Fast, especially in the high-load regime. The DSTMR degrades even faster because its conservative reservation-based routing exhausts available resources early under increasing load. The results show that CRT-Fast improves schedulability more effectively than strict conflict avoidance, rigid shortest-path routing, greedy dynamic routing, or conservative reservation-based scheduling.

\begin{figure}[htbp]
	\centering
	\setlength{\belowcaptionskip}{-0.2cm} 
	\includegraphics[width=0.4\textwidth]{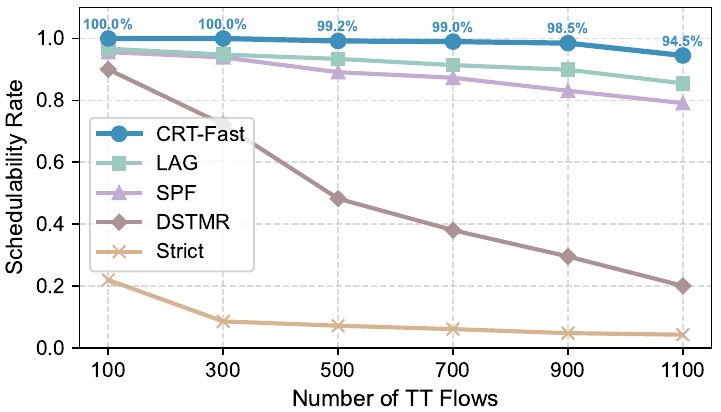}
	\caption{Average flow scheduling success rate.}
	\label{fig:performance}
\end{figure}

\begin{figure}[!t]
	\centering
	\setlength{\belowcaptionskip}{-0.2cm} 
	\subfloat[][\footnotesize{Iridium}]{
		\includegraphics[width=0.47\columnwidth]{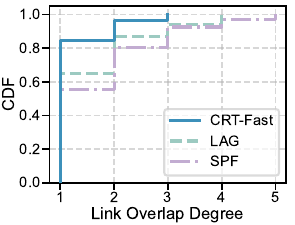}
		\label{fig:cdf_i}
	}
	\subfloat[][\footnotesize{Starlink}]{
		\includegraphics[width=0.47\columnwidth]{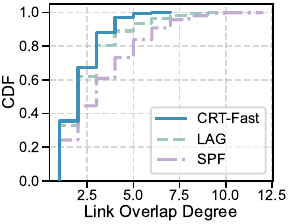}
		\label{fig:cdf_s}
	}
	\caption{CDF of link overlap degrees ($n_e$).}
	\label{fig:cdf_analysis}
\end{figure}

\begin{figure}[!t]
	\centering
	\setlength{\belowcaptionskip}{-0.2cm} 
	\subfloat[][\footnotesize{Iridium}]{
		\includegraphics[width=0.47\columnwidth]{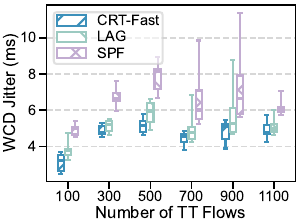}
		\label{fig:jitter_i}
	}
	\subfloat[][\footnotesize{Starlink}]{
		\includegraphics[width=0.47\columnwidth]{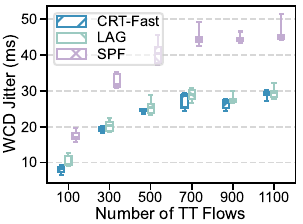}
		\label{fig:jitter_s}
	}
	\caption{Distribution of collision-induced jitter.}
	\label{fig:jitter_analysis}
\end{figure}

\textbf{Conflict Mitigation and Jitter Analysis:}
Fig. \ref{fig:cdf_analysis} and Fig. \ref{fig:jitter_analysis} examine the relationship between link overlap and collision-induced jitter. Taking Iridium as an example (Fig. \ref{fig:cdf_i}), CRT-Fast yields the most concentrated overlap distribution: more than 85\% of links remain at $n_e=1$, and the maximum overlap is limited to 3. In comparison, SPF and LAG exhibit longer tails, with overlap degrees reaching 5 and 4, respectively. This indicates that CRT-Fast is more effective at suppressing local collision intensity. Fig. \ref{fig:jitter_i} shows that the jitter distribution follows the same trend. Because the WCD grows with the overlap degree, the tighter overlap control achieved by CRT-Fast leads to lower median jitter and smaller variance. These results confirm that iterative layering can reduce collision-induced jitter while preserving high schedulability. Since Strict Non-Overlapping enforces fully conflict-free routing, its overlap degree and collision-induced jitter are identically zero.

\textbf{Link-Induced Jitter Analysis:}
To isolate link-delay-induced jitter from collision-induced jitter, we construct a conflict-free traffic setting with 200 TT flows in the Iridium constellation, where no two flows share transmission links in the same slot. Under this setting, the accumulated WCD is zero, and the remaining delay variation is caused mainly by dynamic link delays and route changes across slots. As shown in Fig.~\ref{fig:LinkDelayJitter}, CRT-Fast keeps the link-induced jitter nearly zero across all tested slot lengths, while others exhibit substantially larger jitter and wider distributions.

\begin{figure}[!t]
	\setlength{\belowcaptionskip}{-0.2cm} 
	\centering
	\includegraphics[width=0.42\textwidth]{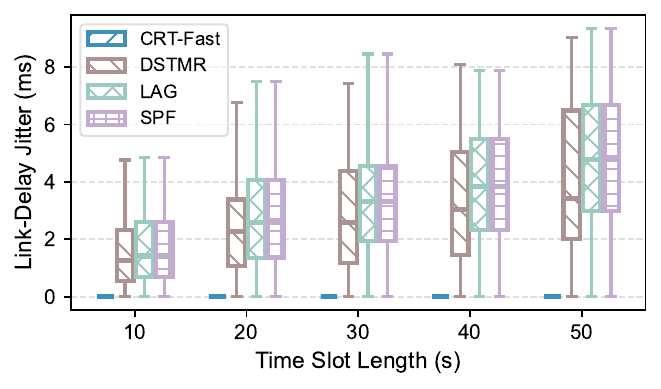}
	\caption{Distribution of link-induced jitter under conflict-free scheduling.}
	\label{fig:LinkDelayJitter}
\end{figure}

\begin{figure}[!t]
	\setlength{\belowcaptionskip}{-0.2cm} 
	\centering
	\includegraphics[width=0.42\textwidth]
	{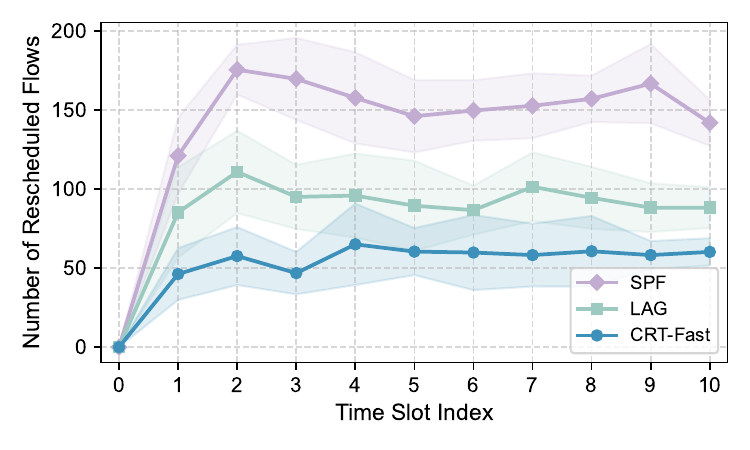}
	\caption{Path stability analysis under dynamic LEO topology handovers.}
	\label{fig:Handover_Stability}
\end{figure}

\textbf{Path Stability under LEO Handovers:}
To evaluate robustness against topology dynamics, we simulated 11 consecutive time slots with 400 flows on Iridium. In each slot, 3\% of links were randomly disconnected and 15\% experienced delay perturbations. We measure stability by the number of rescheduled flows. As shown in  Fig.~\ref{fig:Handover_Stability}, CRT-Fast requires the fewest path updates on average, with 57.3 rescheduled flows per slot. The shaded region indicates $\pm1$ standard deviation. This improvement is mainly due to the path continuity mechanism, which first attempts to reuse the previous-slot path whenever it remains valid. In contrast, LAG and SPF independently re-optimize routing in each slot, making them more sensitive to local topology changes. CRT-Fast reduces rescheduling by about 38.7\% compared with LAG and 62.7\% compared with SPF, indicating stronger path stability under frequent LEO handovers.

\begin{figure}[!t]
	\centering
	\setlength{\belowcaptionskip}{-0.2cm} 
	\includegraphics[width=0.45 \textwidth]{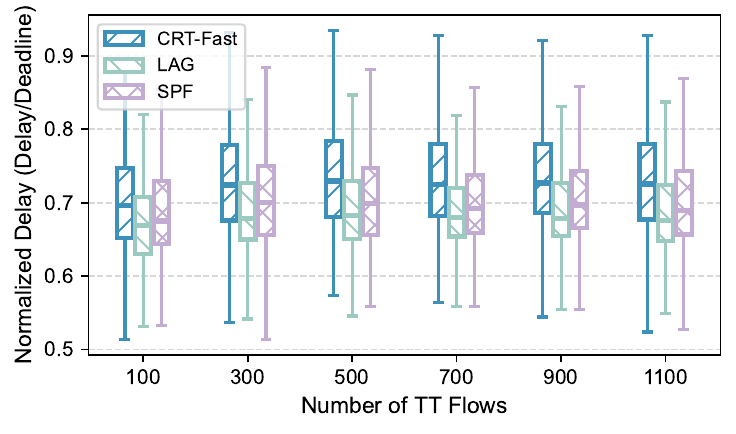}
	\caption{Normalized End-to-End Delay Performance.}
	\label{fig:delayperformance}
\end{figure}

\textbf{End-to-End Delay Performance:}
Fig.~\ref{fig:delayperformance} reports the normalized e2e delay of scheduled flows on Iridium. Compared with SPF and LAG, CRT-Fast exhibits a slightly wider delay range. This increase does not indicate degraded feasibility, rather, it reflects that CRT-Fast admits more hard-to-schedule flows that require longer paths or traverse more congested regions of the network. In contrast, SPF and LAG tend to reject such flows earlier, which leads to a lower average delay among only the schedulable flows. Despite the marginal increase in normalized delay, all flows scheduled by CRT-Fast remain within their deadlines. These results indicate that CRT-Fast improves schedulability while effectively utilizing the available delay slack.

\textbf{Impact of End-to-End Deadline on Scheduling:} We fix the number of flows at 1,000 and vary the e2e deadline in both Iridium and Starlink constellations. As shown in Fig.~\ref{fig:ddl_i} for Iridium and Fig.~\ref{fig:ddl_s} for Starlink, schedulability is improved as the deadline becomes more relaxed. In Iridium, the scheduling success rate increases from 32.4\% at 32 ms to 98.0\% at 128 ms, and reaches 100\% at 320 ms and above. In Starlink, it increases from 63.8\% at 320 ms to 98.3\% at 512 ms, and reaches 100\% at 576 ms and above. These results show that tighter deadlines significantly restrict feasible routing and residence-time allocation, while more relaxed deadlines provide greater scheduling flexibility.

\begin{figure}[t!]
	\centering
	\subfloat[][\footnotesize{Iridium}]{
		\includegraphics[width=0.48\columnwidth]{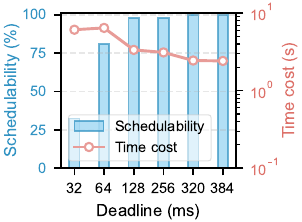}
		\label{fig:ddl_i}
	}
	\subfloat[][\footnotesize{Starlink}]{
		\setlength{\belowcaptionskip}{-0.2cm} 
		\includegraphics[width=0.48\columnwidth]{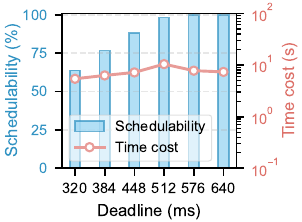}
		\label{fig:ddl_s}
	}
	\caption{Performance of CRT-Fast with different e2e deadline.}
	\label{fig:ddl_analysis}
\end{figure}

\textbf{Scalability Analysis:}
We stress-test CRT-Fast on Starlink by increasing the number of flows from 1,000 to 10,000. As shown in Fig.~\ref{fig:Schedulability}, CRT-Fast maintains strong schedulability over a wide range of large-scale traffic loads, achieving 98.2\% at 1,000 flows and still admitting 66.7\% of flows at 10,000. This result demonstrates that the proposed collision-tolerant scheduling framework remains effective even in very large and dense LEO scheduling scenarios. Meanwhile, the algorithm successfully scales to all tested problem sizes, showing that CRT-Fast can handle large candidate sets and complex topology-constrained scheduling instances on a mega-constellation.

\begin{figure}[!t]
	\centering
	\includegraphics[width=0.45\textwidth]{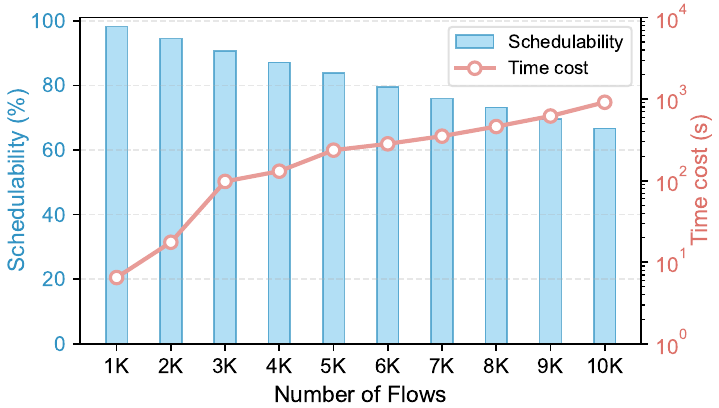}
	\caption{Scalability performance of the CRT-Fast.}
	\label{fig:Schedulability}
\end{figure}

\section{Related Work}

\textbf{Terrestrial deterministic transmission:}
Deterministic transmission in terrestrial networks has been studied primarily in the context of TSN. Early work formulates routing and scheduling for time-triggered traffic as constraint-solving problems, such as SMT, ILP, or related optimization models, in order to compute feasible or optimal schedules~\cite{Atallah2020Routing,schweissguth2017ilp}. To improve support for traffic with timing guarantees, IEEE 802.1Qch introduces Cyclic Queueing and Forwarding (CQF)~\cite{8021qch}, and subsequent studies further refine CQF-based mechanisms for better delay control and resource utilization~\cite{yang2023burst,Yan2020injection}. More recently, learning-based methods such as DeepScheduler and TTDeep explore deep reinforcement learning to accelerate scheduling and improve schedulability without relying entirely on handcrafted rules~\cite{He2023deep,jia2021ttdeep}.

Although these methods are effective in wired TSN environments, they are generally developed under two assumptions: accurate network-wide time synchronization and relatively stable topologies. These assumptions are difficult to satisfy in LEO satellite networks, where links are time-varying, handovers are frequent, and path availability changes continuously. Therefore, terrestrial deterministic transmission techniques cannot be directly applied to LEO scenarios.

\textbf{Deterministic transmission in LEO satellite networks:}
Existing studies on time-sensitive transmission in LEO satellite networks mainly focus on delay reduction, routing adaptation, and reliability enhancement. For example, CPF~\cite{wang2023cpf} and TSN-MCQ~\cite{ma2023time} improve queueing and forwarding mechanisms to provide bounded service delay for time-sensitive flows. Other studies investigate delay-aware routing strategies for dynamic satellite topologies~\cite{dong2023delay,geng2020optimal}. To improve robustness under topology variations and link failures, FastTS~\cite{peng2024fastts} introduces a fault-tolerant heuristic scheduling strategy, while Lai et al.~\cite{lai2023achieving} develop a resilient routing mechanism that improves route restoration under topology fluctuations.

Above all, the existing LEO-oriented approaches still leave two issues insufficiently addressed. First, they mainly focus on delay reduction or routing robustness, but explicitly controlling collision-induced delay jitter. Second, they usually do not model the effect of imperfect time synchronization across different source clock domains, which is a fundamental challenge in asynchronous LEO environments. In contrast, our work jointly considers topology dynamics, asynchronous collisions, and imperfect synchronization, and develops a residence-time-based collision-tolerant scheduling framework to provide bounded jitter and deterministic transmission.

\section{Conclusion}
We studied deterministic transmission in asynchronous LEO satellite networks and proposed CRT, a deterministic transmission framework that operates without requiring global time synchronization. By introducing a residence-time mechanism based on local clocks, CRT compensates for link-delay variations caused by topology dynamics and stabilizes the e2e delay. To address inevitable asynchronous collisions among flows from different sources, we further developed a collision-tolerant scheduling model that improves schedulability while bounding collision-induced jitter under deadline and resource constraints. We showed that the scheduling problem is NP-hard and CRT-Fast, an efficient heuristic algorithm that combines iterative layering with path continuity. Simulations on Iridium and Starlink constellations show that CRT-Fast achieves better schedulability, lower jitter, and stronger path stability than compared baselines. CRT provides a scalable solution for time-sensitive services in dynamic LEO satellite networks.

\bibliographystyle{IEEEtran}
\bibliography{ref}

\begin{thebibliography}{10}
\providecommand{\url}[1]{#1}
\csname url@samestyle\endcsname
\providecommand{\newblock}{\relax}
\providecommand{\bibinfo}[2]{#2}
\providecommand{\BIBentrySTDinterwordspacing}{\spaceskip=0pt\relax}
\providecommand{\BIBentryALTinterwordstretchfactor}{4}
\providecommand{\BIBentryALTinterwordspacing}{\spaceskip=\fontdimen2\font plus
\BIBentryALTinterwordstretchfactor\fontdimen3\font minus
  \fontdimen4\font\relax}
\providecommand{\BIBforeignlanguage}[2]{{%
\expandafter\ifx\csname l@#1\endcsname\relax
\typeout{** WARNING: IEEEtran.bst: No hyphenation pattern has been}%
\typeout{** loaded for the language `#1'. Using the pattern for}%
\typeout{** the default language instead.}%
\else
\language=\csname l@#1\endcsname
\fi
#2}}
\providecommand{\BIBdecl}{\relax}
\BIBdecl

\bibitem{starlink}
F.~Michel, M.~Trevisan, D.~Giordano, and O.~Bonaventure, ``A first look at
  starlink performance,'' in \emph{Proceedings of the 22nd ACM Internet
  Measurement Conference}, 2022, pp. 130--136.

\bibitem{iridium}
S.~R. Pratt, R.~A. Raines, C.~E. Fossa, and M.~A. Temple, ``An operational and
  performance overview of the {IRIDIUM} low earth orbit satellite system,''
  \emph{IEEE Communications Surveys}, vol.~2, no.~2, pp. 2--10, 1999.

\bibitem{Oneweb}
Y.~Henri, ``The oneweb satellite system,'' in \emph{Handbook of Small
  Satellites: Technology, Design, Manufacture, Applications, Economics and
  Regulation}, 2020, pp. 1091--1100.

\bibitem{popovski2022perspective}
P.~Popovski, F.~Chiariotti, K.~Huang, A.~E. Kal{\o}r, M.~Kountouris, N.~Pappas,
  and B.~Soret, ``A perspective on time toward wireless {6G},''
  \emph{Proceedings of the IEEE}, vol. 110, no.~8, pp. 1116--1146, 2022.

\bibitem{yu2023toward}
H.~Yu, T.~Taleb, K.~Samdanis, and J.~Song, ``Toward supporting holographic
  services over deterministic {6G} integrated terrestrial and non-terrestrial
  networks,'' \emph{IEEE Network}, vol.~38, no.~1, pp. 262--271, 2023.

\bibitem{wang2024learning}
Z.~Wang, H.~Yao, T.~Mai, Z.~Li, and C.~P. Chen, ``Learning-driven swarm
  intelligence: Enabling deterministic flows scheduling in {LEO} satellite
  networks,'' \emph{IEEE Transactions on Mobile Computing}, 2024.

\bibitem{sun2025joint}
H.~Sun, H.~Zhang, H.~Ma, and V.~C. Leung, ``Joint scheduling, computing, and
  load balancing for time sensitive traffic in sdn-enabled space-air-ground
  integrated {6G} networks: A federated reinforcement learning approach,''
  \emph{IEEE Transactions on Mobile Computing}, 2025.

\bibitem{xiao2022leo}
Z.~Xiao, J.~Yang, T.~Mao, C.~Xu, R.~Zhang, Z.~Han, and X.-G. Xia, ``{LEO}
  satellite access network ({LEO-SAN}) toward {6G}: Challenges and
  approaches,'' \emph{IEEE Wireless Communications}, vol.~31, no.~2, pp.
  89--96, 2022.

\bibitem{ma2023network}
S.~Ma, Y.~C. Chou, H.~Zhao, L.~Chen, X.~Ma, and J.~Liu, ``Network
  characteristics of {LEO} satellite constellations: A starlink-based
  measurement from end users,'' in \emph{Proceedings of the IEEE INFOCOM
  2023-IEEE Conference on Computer Communications}, 2023, pp. 1--10.

\bibitem{cao2023satcp}
X.~Cao and X.~Zhang, ``Satcp: Link-layer informed {TCP} adaptation for highly
  dynamic {LEO} satellite networks,'' in \emph{Proceedings of the IEEE INFOCOM
  2023 -IEEE Conference on Computer Communications}, 2023, pp. 1--10.

\bibitem{tian2024large}
W.~Tian, C.~Gu, M.~Guo, S.~He, J.~Kang, D.~Niyato, and J.~Chen, ``Large-scale
  deterministic networks: Architecture, enabling technologies, case study, and
  future directions,'' \emph{IEEE Network}, vol.~38, no.~4, pp. 284--291, 2024.

\bibitem{wang2021large}
S.~Wang, B.~Wu, C.~Zhang, Y.~Huang, T.~Huang, and Y.~Liu, ``Large-scale
  deterministic {IP} networks on {CENI},'' in \emph{Proceedings of the IEEE
  INFOCOM 2021-IEEE Conference on Computer Communications Workshops}, 2021, pp.
  1--6.

\bibitem{hu2024time}
Y.~Hu, B.~Guo, C.~Yang, and Z.~Han, ``Time-deterministic networking for
  satellite-based internet-of-things services: Architecture, key technologies,
  and future directions,'' \emph{IEEE Network}, 2024.

\bibitem{finn2018introduction}
N.~Finn, ``Introduction to time-sensitive networking,'' \emph{IEEE
  Communications Standards Magazine}, vol.~2, no.~2, pp. 22--28, 2018.

\bibitem{nasrallah2018ultra}
A.~Nasrallah, A.~S. Thyagaturu, Z.~Alharbi, C.~Wang, X.~Shao, M.~Reisslein, and
  H.~ElBakoury, ``Ultra-low latency ({ULL}) networks: The {IEEE TSN} and {IETF}
  detnet standards and related {5G} ull research,'' \emph{IEEE Communications
  Surveys \& Tutorials}, vol.~21, no.~1, pp. 88--145, 2018.

\bibitem{kopetz2003time}
H.~Kopetz and G.~Bauer, ``The time-triggered architecture,'' \emph{Proceedings
  of the IEEE}, vol.~91, no.~1, pp. 112--126, 2003.

\bibitem{IEEE8021AS}
{IEEE Std 802.1AS}, ``{IEEE Standard for Local and Metropolitan Area
  Networks---Timing and Synchronization for Time-Sensitive Applications},''
  2020.

\bibitem{IEEE1588}
{IEEE Std 1588}, ``{IEEE Standard for a Precision Clock Synchronization
  Protocol for Networked Measurement and Control Systems},'' Jun. 2020.

\bibitem{IEEE8021Qbv}
{IEEE Std 802.1Qbv}, ``{IEEE Standard for Local and Metropolitan Area
  Networks---Bridges and Bridged Networks---Amendment 25: Enhancements for
  Scheduled Traffic},'' Mar. 2016.

\bibitem{pan2023measuring}
J.~Pan, J.~Zhao, and L.~Cai, ``Measuring a low-earth-orbit satellite network,''
  in \emph{Proceedings of the IEEE 34th Annual International Symposium on
  Personal, Indoor and Mobile Radio Communications}, 2023, pp. 1--6.

\bibitem{chen2024asynchronous}
X.~Chen and Z.~Luo, ``Asynchronous interference mitigation for leo
  multi-satellite cooperative systems,'' \emph{IEEE Transactions on Wireless
  Communications}, 2024.

\bibitem{li2018time}
Z.~Li, H.~Wan, Y.~Deng, X.~Zhao, Y.~Gao, X.~Song, and M.~Gu, ``Time-triggered
  switch-memory-switch architecture for time-sensitive networking switches,''
  \emph{IEEE Transactions on Computer-Aided Design of Integrated Circuits and
  Systems}, vol.~39, no.~1, pp. 185--198, 2018.

\bibitem{IEEE8021Qbu}
{IEEE Std 802.1Qbu}, ``{IEEE Standard for Local and Metropolitan Area
  Networks---Bridges and Bridged Networks---Amendment 26: Frame Preemption},''
  2016.

\bibitem{vygen1995np}
J.~Vygen, ``Np-completeness of some edge-disjoint paths problems,''
  \emph{Discrete Applied Mathematics}, vol.~61, no.~1, pp. 83--90, 1995.

\bibitem{yen1971finding}
J.~Y. Yen, ``Finding the k shortest loopless paths in a network,''
  \emph{management Science}, vol.~17, no.~11, pp. 712--716, 1971.

\bibitem{jiang2023spatio}
X.~Jiang, Y.~Huang, J.~Li, H.~He, S.~Chen, F.~Yang, and J.~Yang,
  ``Spatio-temporal routing, redundant coding and multipath scheduling for
  deterministic satellite network transmission,'' \emph{IEEE Transactions on
  Communications}, vol.~71, no.~5, pp. 2860--2875, 2023.

\bibitem{Atallah2020Routing}
A.~A. Atallah, G.~B. Hamad, and O.~A. Mohamed, ``Routing and scheduling of
  time-triggered traffic in time-sensitive networks,'' \emph{IEEE Transactions
  on Industrial Informatics}, vol.~16, no.~7, pp. 4525--4534, 2020.

\bibitem{schweissguth2017ilp}
E.~Schweissguth, P.~Danielis, D.~Timmermann, H.~Parzyjegla, and G.~M{\"u}hl,
  ``{ILP}-based joint routing and scheduling for time-triggered networks,'' in
  \emph{Proceedings of the 25th International Conference on Real-Time Networks
  and Systems}, 2017, pp. 8--17.

\bibitem{8021qch}
{IEEE Std 802.1Qch}, ``{IEEE Standard for Local and Metropolitan Area
  Networks---Bridges and Bridged Networks---Amendment: Cyclic Queuing and
  Forwarding},'' 2017.

\bibitem{yang2023burst}
D.~Yang, Z.~Cheng, W.~Zhang, H.~Zhang, and X.~Shen, ``Burst-aware
  time-triggered flow scheduling with enhanced multi-{CQF} in time-sensitive
  networks,'' \emph{IEEE/ACM Transactions on Networking}, vol.~31, no.~6, pp.
  2809--2824, 2023.

\bibitem{Yan2020injection}
J.~Yan, W.~Quan, X.~Jiang, and Z.~Sun, ``Injection time planning: Making {CQF}
  practical in time-sensitive networking,'' in \emph{Proceedings of the IEEE
  INFOCOM 2020 -IEEE Conference on Computer Communications}, 2020, pp.
  616--625.

\bibitem{He2023deep}
X.~He, X.~Zhuge, F.~Dang, W.~Xu, and Z.~Yang, ``Deepscheduler: Enabling
  flow-aware scheduling in time-sensitive networking,'' in \emph{Proceedings of
  the IEEE INFOCOM 2023 -IEEE Conference on Computer Communications}, 2023, pp.
  1--10.

\bibitem{jia2021ttdeep}
H.~Jia, Y.~Jiang, C.~Zhong, H.~Wan, and X.~Zhao, ``Ttdeep: Time-triggered
  scheduling for real-time ethernet via deep reinforcement learning,'' in
  \emph{Proceedings of the IEEE Global Communications Conference}, 2021, pp.
  1--6.

\bibitem{wang2023cpf}
F.~Wang, D.~Wu, W.~He, Z.~Li, Q.~Zhang, and H.~Yao, ``{CPF}: Bridging
  time-sensitive networks into large-scale {LEO} satellite networks,'' in
  \emph{Proceedings of the International Wireless Communications and Mobile
  Computing}, 2023, pp. 1--6.

\bibitem{ma2023time}
X.~Ma, S.~Li, Z.~Guan, J.~Li, H.~Sun, Y.~Wang, and H.~Guo, ``Time-sensitive
  networking mechanism aided by multilevel cyclic queues in {LEO} satellite
  networks,'' \emph{Electronics}, vol.~12, no.~6, p. 1357, 2023.

\bibitem{dong2023delay}
F.~Dong, Y.~Zhang, G.~Liu, H.~Yu, and C.~Sun, ``Delay-sensitive service
  provisioning in software-defined low-earth-orbit satellite networks,''
  \emph{Electronics}, vol.~12, no.~16, p. 3474, 2023.

\bibitem{geng2020optimal}
S.~Geng, S.~Liu, Z.~Fang, and S.~Gao, ``An optimal delay routing algorithm
  considering delay variation in the {LEO} satellite communication network,''
  \emph{Computer Networks}, vol. 173, p. 107166, 2020.

\bibitem{peng2024fastts}
G.~Peng, S.~Wang, T.~Huang, F.~Li, K.~Zhao, Y.~Huang, and Z.~Xiong, ``{FastTS}:
  Enabling fault-tolerant and time-sensitive scheduling in space-terrestrial
  integrated networks,'' \emph{IEEE Journal on Selected Areas in
  Communications}, vol.~42, no.~12, pp. 3551--3565, 2024.

\bibitem{lai2023achieving}
Z.~Lai, H.~Li, Y.~Wang, Q.~Wu, Y.~Deng, J.~Liu, Y.~Li, and J.~Wu, ``Achieving
  resilient and performance-guaranteed routing in space-terrestrial integrated
  networks,'' in \emph{Proceedings of the IEEE INFOCOM 2023 -IEEE Conference on
  Computer Communications}, 2023, pp. 1--10.

\end{thebibliography}

\end{document}